\documentclass[twocolumn]{article}
\usepackage{amsmath,amssymb,amsthm}

\def\BibTeX{{\rm B\kern-.05em{\sc i\kern-.025em b}\kern-.08em
    T\kern-.1667em\lower.7ex\hbox{E}\kern-.125emX}}

\newtheorem{theo}{Theorem}
\newtheorem{propo}[theo]{Proposition}
\newcommand{\itl}{\textit}

\def\fpd#1#2{{\displaystyle\frac{\partial #1}{\partial #2}}}
\def\spd#1#2#3{{\displaystyle\frac{\partial^2 #1}
{\partial #2\partial #3}}}
\def\onehalf{{\frac12}}

\usepackage{amssymb}
\usepackage{graphicx}

\usepackage{pslatex}

\begin{document}

\title{Hamiltonization and geometric integration of nonholonomic mechanical systems}

\author{T.\ Mestdag$^1$, A.M.\ Bloch$^2$ and O.E.\ Fernandez$^2$\\
$^1$Department of Mathematical Physics and Astronomy, Ghent University\\
Krijgslaan 281, S9, 9000 Gent, Belgium\\
$^2$Department of Mathematics, University of Michigan\\
530 Church Street, Ann Arbor, MI-48109, USA \\
email: tom.mestdag@ugent.be, abloch@umich.edu, oscarum@umich.edu}

\maketitle
\thispagestyle{empty}\pagestyle{empty}

\begin{abstract}
\small{\textbf{In this paper we study a Hamiltonization procedure for mechanical systems with velocity-depending (nonholonomic) constraints. We first rewrite the nonholonomic equations of motion as Euler-Lagrange equations, with a Lagrangian that follows from rephrasing the issue in terms of the inverse problem of Lagrangian mechanics. Second, the Legendre transformation transforms the Lagrangian in the sought-for Hamiltonian. As an application, we compare some variational integrators for the new Lagrangians with some known nonholonomic integrators. }}
\end{abstract}


\section{Introduction}

{M}{any} interesting mechanical systems are subject to
additional velocity-dependent (i.e.\ nonholonomic) constraints.
Typical engineering problems that involve such constraints arise for
example in robotics, where the wheels of a mobile robot are often
required to roll without slipping, or where one is
interested in guiding the motion of a cutting tool.

The direct motivation for our paper \cite{paper} was to be found in
interesting results that appeared in \cite{QB}, where
the authors propose a way to quantize some of the well-known
classical examples of nonholonomic systems. On the way to
quantization, the authors propose an alternative Hamiltonian
representation for those nonholonomic systems.  However, the
``Hamiltonization'' method introduced in \cite{QB} can only be
applied to systems for which the solutions are already known
explicitly.

Nonholonomic systems have a more natural description in the Lagrangian framework. In \cite{paper}, we explained how one can associate
to the nonholonomic equations of motion a family of systems of
second-order ordinary differential equations and we applied the conditions of the
inverse problem of the calculus of variations on those associated
systems to search for the existence of a regular Lagrangian. (The inverse problem of the calculus of variations deals with the question of whether or not a given system of second-order differential equations is equivalent with the Euler-Lagrange equations of a yet to be determined regular Lagrangian, see e.g.\ \cite{Santilli}). If such an
unconstrained regular Lagrangian exists for one of the associated
systems, we can always find an associated Hamiltonian by means of
the Legendre transformation. Since our method only made use of the
equations of motion of the system it did not depend on the knowledge
of its explicit solutions.

 A system for which no exact solutions are known can only be
integrated by means of numerical methods. In addition to the above
mentioned application to quantization, our Hamiltonization method
may also be useful from this point of view. Numerical integrators that preserve the underlying geometric structure of a system are called geometric integrators. A geometric integrator
of a Lagrangian system uses a discrete Lagrangian that resembles as
much as possible the continuous Lagrangian (see e.g.\ \cite{West}).
On the other hand, the succes of a so-called nonholomic integrator (see e.g.\
\cite{CortesMartinez,FedZen}) relies not only on the choice of
 a discrete Lagrangian but also on the choice of a discrete version of the constraint manifold. It seems therefore
reasonable that if a free Lagrangian for the nonholonomic system
exists, the Lagrangian integrator may perform better than a
nonholonomic integrator with badly chosen discrete constraints.

In the next section we recall the set-up and the main results of our paper \cite{paper}. In section~\ref{sec3} we compare some nonholonomic and variational geometric integrators for a few of the classical nonholonomic systems. In section~\ref{sec4}, we indicate some ideas on how we wish to extend the results of this paper.

\section{A class of nonholonomic systems}\label{sec2}

We will consider only a certain class of nonholonomic systems on
${\mathbb R}^n$: We will assume that the configuration space of the system is a space with co\"ordinates $(r_1,r_2,s_\alpha)$, that the Lagrangian of the system is
given by the function \begin{equation}\label{nonholLag} L=\frac{1}{2}(I_1{\dot r_{1}}^2+I_2{\dot r_{2}}^2 +
\sum_\alpha I_\alpha {\dot s}_\alpha^2)\end{equation} and that the nonholonomic constraints
are all of the form \begin{equation}\label{con} {\dot s}_\alpha =-A_\alpha(r_{1})\dot r_{2}.\end{equation} The
nonholonomic equations of motion follow from d`Alembert's priciple (see e.g.\ \cite{Bloch}). For systems in our class they are given by the equations
\[\left\{ \begin{array}{l}
\displaystyle \frac{d}{dt}\Big(\fpd{L}{{\dot r}_1}\Big) - \fpd{L}{r_1} =
0,\\[3mm] \displaystyle \frac{d}{dt}\Big(\fpd{L}{{\dot r}_2}\Big) - \fpd{L}{r_2} =
\lambda_\alpha A_\alpha, \\[3mm]\displaystyle \frac{d}{dt}\Big(\fpd{L}{{\dot s}_\alpha}\Big) - \fpd{L}{s_\alpha} =
\lambda_\alpha,\end{array}\right.
\]
together with the constraint equations (\ref{con}). After eliminating the Lagrange multipliers by means of the constraints, one gets
\begin{equation}\label{nonhol} \left\{ \begin{array}{l} \ddot r_{1} =0,\\[2mm] \ddot r_{2} =
-N^2 K \dot r_{1}\dot r_{2},\\[2mm] {\dot s}_\alpha =-A_\alpha\dot
r_{2}, \end{array}\right.
\end{equation} where $ N(r_{1})=(I_2+\sum_\alpha I_\alpha
A_\alpha^2)^{-\frac{1}{2}}$ is related to the invariant measure of
the system and $K=\sum_\beta I_\beta A_\beta A'_\beta$. with $A'_\beta = \partial_{r_1}A_\beta$.

Some basic examples of nonholonomic systems that lie in this class
are the following ones. The classic example of a nonholonomically constrained free particle has
a Lagrangian and constraint given by
\[ L=\frac{1}{2}\left(\dot{x}^{2} + \dot{y}^{2} +
\dot{z}^{2}\right)\quad\mbox{and}\quad \dot{z} +x\dot{y}=0.\]
A knife edge on a horizontal plane corresponds physically to a blade with
mass $m$ moving in the $xy$ plane at an angle $\phi$ to the
$x$-axis. Its Lagrangian and constraint are given by \[ L =
\frac{1}{2}m(\dot{x}^{2}+\dot{y}^{2})+\frac{1}{2}J\dot{\phi}^{2}\quad\mbox{and}\quad
 \dot{x}\sin(\phi) - \dot{y}\cos(\phi) = 0
.\] Also the vertically rolling disk is an example in our class. The assumption that the disk rolls without slipping over the plane gives rise to nonholonomic constraints. Let $R$ be the radius of the disk. If the triple $(x,y,z=R)$ stands for the co\"ordinates of its centre of mass,
 $\varphi$ for its angle with the $(x,z)$-plane and
$\theta$ for the angle of a fixed line on the disk and a vertical line, then the
nonholomic constraints are of the form
\[ \dot x  = R\cos\varphi \dot\theta\quad \mbox{and}\quad \dot y  =  R
\sin\varphi \dot\theta.\] The Lagrangian of the disk is \[
L = \frac{1}{2}M ({\dot
x}^2 + {\dot y}^2) + \frac{1}{2}I {\dot\theta}^2 + \frac{1}{2}J
{\dot\varphi}^2,\] where $I=\frac{1}{2}MR^2$ and $ J=\frac{1}{4}MR^2$ are the moments of inertia and $M$ is the total mass of the disk. For the vertically rolling disk $N$ is a constant and $K=0$.
\begin{center}
\begin{figure}[h]\hspace*{1cm}\includegraphics[scale=0.7]{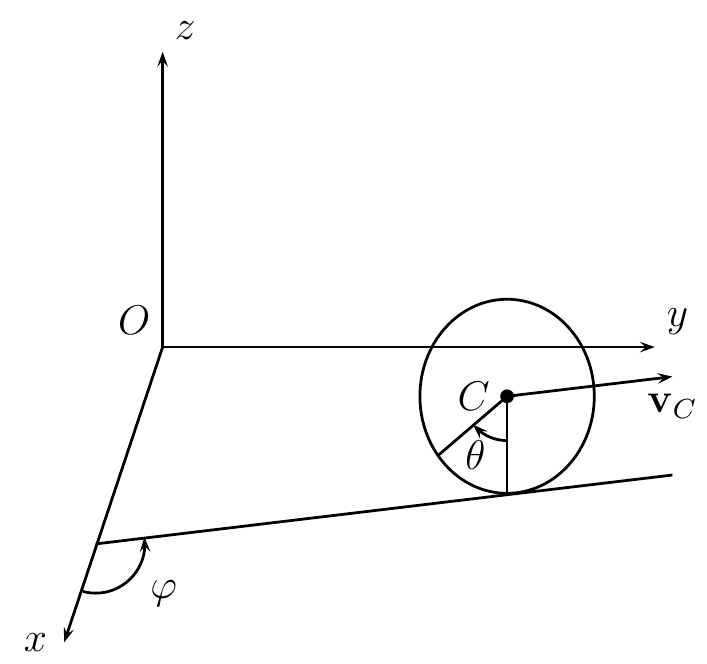}\caption{The vertically rolling disk}
\end{figure}
\end{center}
Finally, also the examples of the mobile robot with fixed orientation and the two-wheeled carriage (see e.g. \cite{FB}) lie within our class.

The equations of motion (\ref{nonhol}) are a mixed set of first- and
second-order differential equations. On the other hand, the
Euler-Lagrange equations
\[
 \frac{d}{dt}\left( \fpd{\tilde L}{{\dot q}^i}
\right)-\fpd{\tilde L}{q^i}=0
\]
of a regular Lagrangian $\tilde L$ are second-order
differential equations (only) [The tilde in $\tilde L$ will always denote that the Lagrangian is free, and that it should not be confused with the original Lagrangian $L$ of the nonholonomic system.]. We therefore need a way to associate a
second-order system to our nonholonomic system. One possible choice
of doing so is the system
\begin{equation} \label{first} \left\{ \begin{array}{l}
\ddot r_{1} =0,\\[2mm] \ddot r_{2} = -N^2 K \dot r_{1}\dot r_{2},\\[2mm] {\ddot
s}_\alpha = -\Big(A'_{\alpha} - N^2 K A_\alpha \Big) \dot r_{1} \dot
r_{2}. \end{array}\right.\end{equation} The above second-order system has the property that its solution set  contains, among other, also the solutions of the nonholonomic
dynamics (\ref{nonhol}) when restricted to the constraints. Another
choice for an `associated system' with the same property is e.g.\
\begin{equation}\label{second}  \left\{ \begin{array}{l}
\hspace*{-2mm} \ddot r_{1} =0,\\[2mm] \ddot r_{2} = -N^2 K \dot r_{1}\dot
 r_{2},\\[2mm] \displaystyle
 {\ddot s}_\alpha =  (A'_{\alpha} - N^2 K A_\alpha) \dot r_{1} \left(\frac{{\dot
s}_\alpha}{A_\alpha}\right) \end{array}\right.    \end{equation}(no sum over $\alpha$). It is clear, that there are in fact an infinite number of such associated second-order systems, but we will concentrate in this paper on the above two. For
some other possible choices, see \cite{paper}.

\begin{propo} {\em \label{Prop1} 1.\  There does not exist a regular Lagrangian whose Euler-Lagrange
equations are equivalent with the second-order system (\ref{first}) (for the classical examples cited above). \newline 2.\ The Euler-Lagrange equations of the Lagrangian \begin{equation} \label{Lag1} \tilde L = \frac{1}{2}I_1{\dot r_{1}}^2 + \frac{1}{2N} \left( C_2 \frac{{\dot
r_{2}}^2}{\dot r_{1}} +\sum_\beta C_\beta \frac{{\dot
s}_\beta^2}{A_\beta \dot r_{1}}\right)\quad (C_\alpha\neq 0) \end{equation} are
equivalent with the second-order system (\ref{second}). If the
invariant measure density $N$ is a constant, then also\begin{equation} \label{Lag2} \tilde L =
\frac{1}{2}I_1{\dot r_{1}}^2 + \frac{1}{2} I_2 {\dot r}_2^2 +
\frac{1}{2 N} \sum_\beta a_\beta \frac{{\dot s}_\beta^2}{A_\beta
\dot r_{1}}\quad (C_\alpha\neq 0)\end{equation} is a regular
Lagrangian for the system (\ref{second}).}
\end{propo}

\begin{proof} We give here only an outline of the method we've used to prove the statements. For full details, see \cite{paper}. Assume we are given a
system of second-order ordinary differential equations \[{\ddot
q}^i=f^i(q,\dot q).\] The search for a regular Lagrangian
 is known in the literature as `the inverse problem of the calculus of variations', and has a
 long history (for a recent survey on this history, see e.g. \cite{KP} and the long list of references therein). In order for a regular Lagrangian
$\tilde L(q,\dot{q})$ to exist we must be able to find functions
$g_{ij}(q,\dot q)$, so-called multipliers, such that
\[
g_{ij}({\ddot q}^j-f^j) = \frac{d}{dt}\left( \fpd{\tilde L}{{\dot q}^i}
\right)-\fpd{\tilde L}{q^i}.
\]
It can be shown \cite{Douglas,Santilli} that the multipliers must
satisfy
\begin{eqnarray*}
&&\det(g_{ij})\neq 0,\quad\quad g_{ji}=g_{ij},\quad\quad
\fpd{g_{ij}}{{\dot q}^k}=\fpd{g_{ik}}{{\dot q}^j};\\&&
\Gamma(g_{ij}) - \nabla^k_j g_{ik}- \nabla^k_i g_{kj}=0,\\ &&
g_{ik}\Phi^k_j = g_{jk}\Phi^k_i;
\end{eqnarray*}
where \[\nabla^i_j = -\onehalf
\partial_{{\dot q}^j}f^i\] and
\[
\Phi^k_j = \Gamma\left(\partial_{{\dot
q}^j}{f^k}\right)-2\partial_{q^j}{f^k}-\onehalf\partial_{{\dot
q}^j}{f^l}\partial_{{\dot q}^l}{f^k}.
\]
The symbol $\Gamma$ stands for the vector field ${\dot
q}^i\partial_{q^i} + f^i
\partial_{{\dot q}^i}$ that can naturally be associated to the system ${\ddot
q}^i=f^i(q,\dot q)$. Conversely, if one can find functions $g_{ij}$
satisfying these conditions then the equations $\ddot{q}^i=f^i$ are
derivable from a regular Lagrangian. Moreover, if a regular
Lagrangian $\tilde L$ can be found, then its Hessian $\spd{\tilde L}{{\dot
q}^i}{{\dot q}^j}$ is a multiplier.

The above conditions are generally referred to as the Helmholtz
conditions. They are a mixed set of coupled algebraic and
PDE conditions in $(g_{ij})$. We will refer to the penultimate
condition as the `$\nabla$- condition,' and to the last one as the
`$\Phi$-condition.' The algebraic $\Phi$-conditions are of course
the most interesting to start from. In fact, we can easily derive
more algebraic conditions (see e.g.\ \cite{Towards}). For example,
by taking a $\Gamma$-derivative of the $\Phi$-condition, and by
replacing $\Gamma(g_{ij})$ everywhere by means of the
$\nabla$-condition, we arrive at a new algebraic condition of
the form
\[
g_{ik}(\nabla\Phi)^k_j = g_{jk}(\nabla\Phi)^k_i,
\]
where \[(\nabla\Phi)^i_j = \Gamma(\Phi^i_j)  -
\nabla^i_m\Phi^m_j-\nabla^m_j\Phi^i_m.\] This $(\nabla\Phi)$-condition will,
of course, only give new information as long as it is independent
from the $\Phi$-condition (this will not be the case, for example,
if the commutator of matrices $[\Phi,\nabla\Phi]$ vanishes). One can
repeat the above process on the $(\nabla\Phi)$-condition, and so on
to obtain possibly independent
$(\nabla\ldots\nabla\Phi)$-conditions.
A second route to additional algebraic conditions arises from the
derivatives of the $\Phi$-equation in $\dot q$-directions. One can
sum up those derived relations in such a way that the terms in
 $\partial_{{\dot q}_k}g_{ij}$ disappear on account of the symmetry in all their indices.
The new algebraic relation in $g_{ij}$ is then of the form
\[
g_{ij}R^j_{kl} + g_{lj}R^j_{ik} + g_{kj}R^j_{li}= 0,
\]
where \[R^j_{kl}= \partial_{{\dot q}^j}(\Phi^k_i)-
\partial_{{\dot q}^i}(\Phi^k_j).\]
As before, this process can
 be continued to obtain more algebraic conditions. Also, any mixture of
the above mentioned two processes leads to possibly new and
independent algebraic conditions. Once we have used up all the
information that we can obtain from this infinite series of
algebraic conditions, we can start looking at the partial
differential equations in the $\nabla$-conditions.

 Let us now come back to the second-order systems (\ref{first}) and (\ref{second}) at hand. The proof of the proposition relies on the fact that for the first systems (\ref{first}), the only matrices $(g_{ij})$ that satisfy the first few algebraic conditions must be non-singular. On the other hand, the two Lagrangians for the system (\ref{second}) follow from an analysis of the Helmholtz conditions with carefully chosen anszatzes.  For more details, see \cite{paper}.
 \end{proof}

Remark that the Lagrangians are not defined for ${\dot r}_1=0$, and we will in general exclude the solutions with that property from the further considerations in this paper.
Any regular Lagrangian system with Lagrangian $\tilde L$ can be transformed into a Hamiltonian one, by making use of the Legendre transformation \[(q^i,{\dot q^i})\mapsto (q^i,p_i = \fpd{\tilde L}{{\dot q}^i}).\] The corresponding Hamiltonian is then \[
\tilde H = p_i q^i -\tilde L.
\] Similarly, the Legendre transformation maps the constraints, viewed as a submanifold of the tangent manifold, onto a submanifold in the cotangent manifold.

%
\begin{propo} {\em
Using the Legendre transformation, the Hamiltonian that corresponds to the Lagrangian (\ref{Lag1})
is given by \[ \tilde H= \frac{1}{2I_1} \left(p_{{1}}+
\frac{1}{2} N \left( \frac{p_{{2}}^2}{C_2} + \sum_\beta A_\beta
\frac{p_\beta^2}{C_\beta} \right)\right)^2. \] The corresponding constraints are
\[
 C_2p_\alpha=-C_\alpha
p_{{2}}.
\] If $N$ is constant, the Hamiltonian that corresponds to the Lagrangian (\ref{Lag2}) is \[ \tilde H= \frac{1}{2I_2} p_2^2 +
\frac{1}{2I_1} \left(p_1+ \frac{1}{2} {N} \left( \sum_\beta
\frac{A_\beta}{a_\beta} {p_\beta^2} \right)\right)^2, \] and the constraints transform into \[ I_2{
{N}{\dot r}_1 p_\alpha} +{a_\alpha} p_2=0,\] where ${\dot
r}_1=(p_1 + \frac{1}{2} {N} \sum_\alpha A_\alpha
p_\alpha^2/a_\alpha)/I_1$.}
\end{propo}

In \cite{paper,CDC} we explain how the above Hamiltonians can be directly derived from Pontryagin's Maximum principle.

\section{Geometric integrators}\label{sec3}

\subsection{Set-up}

As we explained in the introduction, there are now two ways to compute a numeric approximation of a solution of a system in our class: we can use either a nonholonomic
integrator for the original Lagrangian (\ref{nonholLag}) and constraints (\ref{con}) , or we can use a {variational integrator} for one of the Lagrangians (\ref{Lag1}) and (\ref{Lag2}) we have found in Proposition~\ref{Prop1}. Let us come to some details.

Geometric integrators are integrators that preserve the underlying structure of the system. In particular, variational integrators are integrators that are derived from a discrete version of Hamilton's principle. From this discrete variational principle one obtains the so-called discrete Euler-Lagrange equations as follows.
For a mechanical system with Lagrangian
$L$, one needs to choose a discrete Lagrangian $L_d(q_1,q_2)$ (a
function on $Q\times Q$ which resembles as close as possible the
continuous Lagrangian). A solution $q(t)$ is then discretised by an array
$q_k$ which are the solutions of the so-called discrete Euler-Lagrange
equations
\begin{equation}\label{discretevar}
D_1L_d (q_k,q_{k+1}) + D_2 L_d (q_{k-1},q_k)=0.
\end{equation}
These integrators preserve the symplectic and conservative nature of the algorithms.    It is important to realize that a different choice for the discrete Lagrangian may lead to a different geometric integrator. The presence of additional holonomic constraints (i.e.\ `integrable' nonholonomic constraints) can be included by introducing Lagrange's multipliers.

On the other hand, for a nonholonomic integrator of a
nonholonomic system with Lagrangian $L$ and constraints
$\omega^a(q){\dot q}^a=0$, we need to choose both a discrete
Lagrangian $L_d$ and discrete constraint functions
$\omega^a_d$ on $Q\times Q$. The nonholonomic discrete equations are then
\begin{equation}\label{discretenonhol} \left\{ \begin{array}{l}
D_1L_d (q_k,q_{k+1}) + D_2 L_d (q_{k-1},q_k)= (\lambda_k)_a
\omega^a(q_k), \\[2mm] \omega^a_d (q_k,q_{k+1}) =0.\end{array} \right.
\end{equation}
Usually, if $Q$ is a vector space, one takes the discretization in
one of the following ways (for certain $\alpha$ and certain $h$):
\begin{eqnarray} \label{Ld}
L_d (q_1,q_2) &=& L \left(q= (1-\alpha)q_1 + \alpha q_2, {\dot q} =
\frac{q_2-q_1}{h}\right),\\[2mm]
 \omega^a_d (q_1,q_2) &=& \omega^a_i \left(q= (1-\alpha)q_1 + \alpha q_2\right)\frac{q^i_2-q^i_1}{h}.\label{omegad}
\end{eqnarray}
For the rest of the paper, we will concentrate on this discretization procedure. There are, however, many more possibilities to obtain a discrete Lagrangian and discrete constraints. For example, one could take a symmetrized version of the above procedure and use discrete Lagrangians and discrete constraints of the form
\begin{eqnarray*}
L_d (q_1,q_2) &=& \onehalf L \left(q= (1-\alpha)q_1 + \alpha q_2, {\dot q} =
\frac{q_2-q_1}{h}\right) \\ && + \onehalf L \left(q= \alpha q_1 + (1- \alpha) q_2, {\dot q} =
\frac{q_2-q_1}{h}\right),\\[2mm]
 \omega^a_d (q_1,q_2) &=& \onehalf \omega^a_i \left(q= (1-\alpha)q_1 + \alpha q_2\right)\frac{q^i_2-q^i_1}{h}\\ && + \onehalf \omega^a_i \left(q= \alpha q_1 + (1-\alpha) q_2\right)\frac{q^i_2-q^i_1}{h}.
\end{eqnarray*}Also, if the system is invariant under a symmetry group, it is advantageous to construct the integrator in such a way that the discrete system inherits as many as possible of those symmetry properties, see e.g.\ \cite{CortesMartinez}.

The bottom line of the next sections is the following one. If a free Lagrangian for the nonholonomic
system exists, it seems reasonable that the Lagrangian integrator may perform better than a
nonholonomic integrator with badly chosen discrete constraints. In the next sections, we will test this conjecture on a few of the classical examples in our class: the vertically rolling disk, the knife edge and the nonholonomic particle. It will be convenient that for those systems an exact solution of the nonholonomic equations (\ref{nonhol}) is readily available.

\subsection{The vertically rolling disk}

For the vertically rolling disk, we have $(r_1,r_2,s_\alpha) = (\varphi,\theta,x,y)$. It is well-known that the solutions of the nonholonomic equations  with initial conditions $u_\varphi=\dot \varphi(0) \neq 0$ and $u_\theta = \dot\theta(0)$ are all circles with radius $R(u_\theta/u_\phi)$:
\begin{eqnarray}
\theta(t)  &=& u_{\theta}t + \theta_{0},\qquad \varphi(t)  =
u_{\varphi}t + \varphi_{0},\nonumber\\[2mm]
x(t)  &=& \left(\frac{u_{\theta}}{u_{\varphi}}\right)R\sin(\varphi(t)) + x_0,  \nonumber\\[2mm]
y(t)  &=&
-\left(\frac{u_{\theta}}{u_{\varphi}}\right)R\cos(\varphi(t)) + y_0.
\label{vd2}
\end{eqnarray}
Let us put for convenience $M=1$ and $R=1$ and therefore $I=\frac{1}{2}$ and $J = \frac{1}{4}$.
With that the (nonholonomic) Lagrangian and constraints are simply
\[ L = \frac{1}{2} ({\dot
x}^2 + {\dot y}^2) + \frac{1}{4} {\dot\theta}^2 + \frac{1}{8}
{\dot\varphi}^2,\quad   \dot x  = \cos\varphi \dot\theta, \quad  \dot y  =
\sin\varphi \dot\theta.\]
We will first compute the solution of the discrete nonholonomic equations (\ref{discretenonhol}) with the discrete Lagrangian (\ref{Ld}) and the discrete constraints (\ref{omegad}).
Second, since the vertically rolling disk is one of those examples with a constant invariant measure density $N$, we can choose a Lagrangian from the second type (\ref{Lag2}). The simplest choice is probably
\begin{equation}\label{free}
\tilde L =
1/2\left({\dot\varphi}^2+{\dot\theta}^2+ \frac{{\dot x}^2}{\cos(\varphi) {\dot\varphi}} +\frac{{\dot y}^2}{\sin(\varphi)\dot\varphi}\right).
\end{equation}
We now investigate the variational integrator of this Lagrangian, where the discrete Lagrangian is given by (\ref{Ld}). We will fix $h$ (changing it did not have a significant effect) and only concentrate on what happens if we keep $\alpha$ variable.
\begin{figure}[h]
\begin{center}\includegraphics[scale = 0.6]{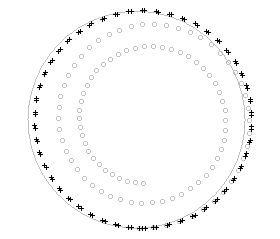} 
\end{center}\caption{Vertically rolling disk with $\alpha=0$.\label{fig1}} \end{figure}
In figure~\ref{fig1} we have plotted the situation for $\alpha = 0$. For a given set of initial positions $(x_0,y_0,\theta_0,\varphi_0,\theta_1,\varphi_1)$ the other initial conditions were chosen in such a way that the solution lies initially on the discrete constraint manifold, i.e.\ in such a way that
\[
x_1= x_0 + \cos\varphi_0 (\theta_1-\theta_0), \quad y_1= y_0 + \sin\varphi_0 (\theta_1-\theta_0).
\]
Unlike the nonholonomic integrator (in grey with circle symbols) the variational integrator (in black with cross symbols) does not show a strong spiral-type solution but a circular path. It is true, however, that the variational solution deviates from the circle predicted by the initial conditions of the solution (\ref{vd2}) (in grey in figure~\ref{fig1}). However, since any circle is determined by 3 of its points, we can find a better match for the circle the variational discrete solution follows by considering the outcome $(x_{k_i},y_{k_i})$ at three different times and by solving the three equations
 \[
 (x_{k_i} - A)^2 +  (x_{k_i} - B)^2 = C^2
 \]
 for $(A,B,C)$. If we do so, we obtain the matching circle (in dots) in figure~\ref{fig2}.
\begin{figure}[h]
\begin{center} \includegraphics[scale=0.5]{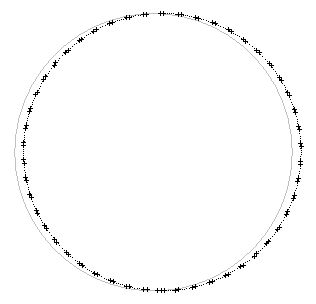}
\end{center}\caption{Vertically rolling disk: circular path.\label{fig2}} \end{figure}

It is well-known that the energy
\[E= \frac{1}{2}({\dot x}^2 + {\dot y}^2) + \frac{1}{4} {\dot\theta}^2 + \frac{1}{8} {\dot\varphi}^2,\] is conserved along the solutions (\ref{vd2}) of the nonholonomic equations of motion. In figure~\ref{fig3} we investigate the performance of the two integrators on the energy function.  The discrete version of the energy is the function we get by substituting, as usual, $1/h(q_k -q_{k-1})$ for ${\dot q}$ in the function $E$ above. The straight line in figure~\ref{fig3} is the energy level predicted by the initial conditions. It is clear that the variational integrator (with crosses) does a better job than the nonholonomic one (with circles).
\begin{figure}[h]
\begin{center}
\includegraphics[scale=0.25]{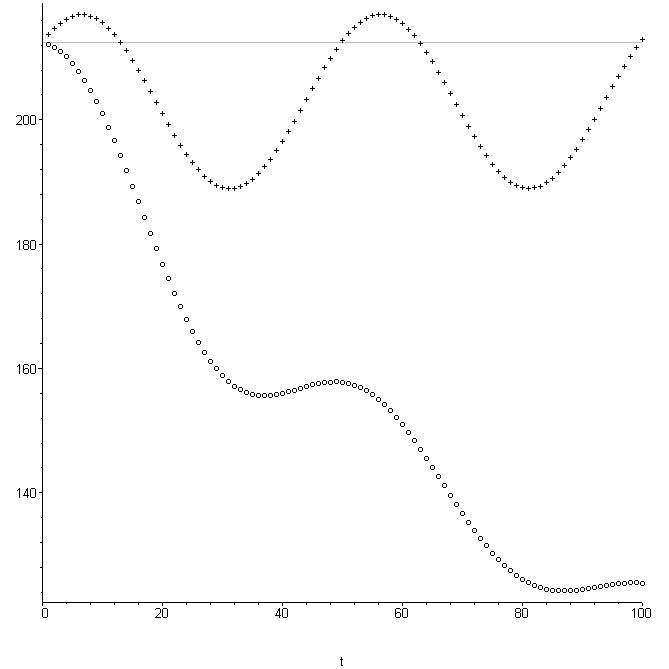} 
\end{center}\caption{Vertically rolling disk: the energy.\label{fig3}}\end{figure}

By construction the nonholonomic integrator conserves the constraints and the variational integrator does not. Indeed, in figure~\ref{fig5} we have plotted the constraint $\dot x - R\cos(\varphi) \dot\theta = 0$. Positive is that, although the variational integrator does not conserve this constraint, it reasonably oscillates around the zero level. Moreover, there is a method to fix this problem. We can introduce a `modified' variational integrator which does conserve the constraints. This integrator considers the
constraints as a constant along the (nonholonomic) motion. That is, it will use the variational discrete Lagrange equations (\ref{discretevar}) for the variables $\theta$ and $\varphi$ (for the free Lagrangian $\tilde L$ given in (\ref{free})), but not the corresponding equations for $x$ and $y$. To get a full system of equations, we  supplemented this with the discrete constraints $\omega^a_d (q_k,q_{k+1}) =0$ which can be written in terms of $x_{k+1}$ and $y_{k+1}$. Figure~\ref{fig4} shows the modified integrator for $\alpha=0$ (with box symbols). The circle in that figure is the one we had before, i.e.\ the one that matches the variational integrator. It shows that the modified integrator has the same circular behaviour as the variational integrator, and on top, it keeps the constraints conserved, see the box symbols on the zero level in figure~\ref{fig5}.
\begin{figure}[h]
\begin{center}
\includegraphics[scale=0.35]{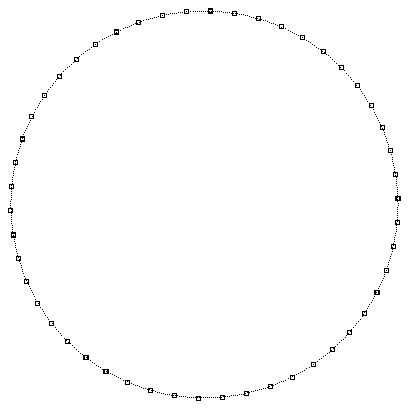} 
\end{center}\caption{Vertically rolling disk: the modified integrator. \label{fig4}}\end{figure}

\begin{figure}[h]
\begin{center}
\includegraphics[scale=0.25]{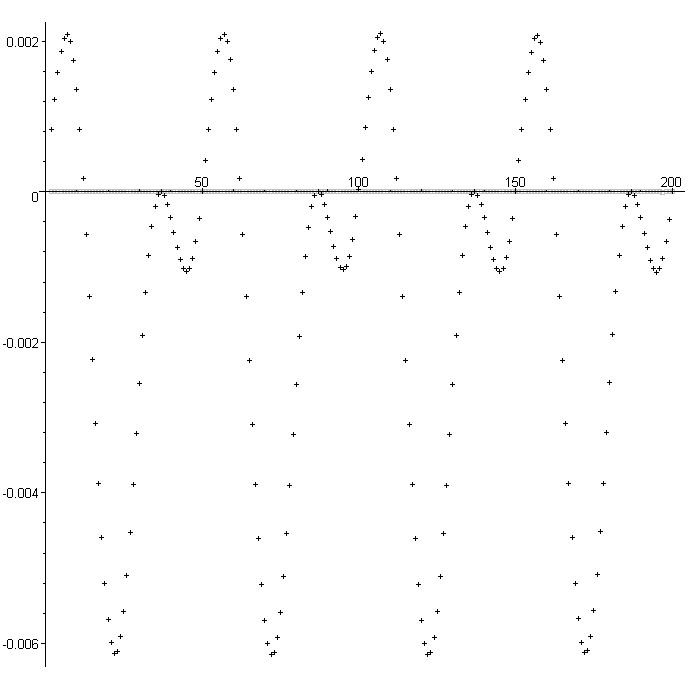} 
\end{center}\caption{Vertically rolling disk: the constraints. \label{fig5}}\end{figure}

Finally, figure~\ref{fig6} shows the effect of changing the parameter $\alpha$. The results for the variational integrator (in black with cross symbols) remain accurate and more or less unchanged. For the nonholonomic integrator (in grey with circle symbols) the effect of changing $\alpha$ is that the inward spiral becomes an outward spiral. At some point (here $\alpha=1/2$) the variational and nonholonomic integrator have the same accuracy.
\begin{figure}[h]
\begin{tabular}{lll}
\begin{minipage}{2.3cm}

\includegraphics[scale=0.4]{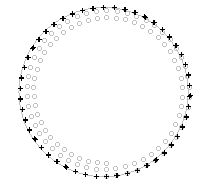} 

\end{minipage}
& \begin{minipage}{2.3cm}

\includegraphics[scale=0.4]{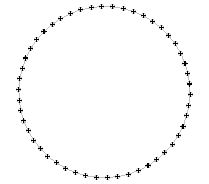} 
\end{minipage} & \begin{minipage}{2.3cm}

\includegraphics[scale=0.4]{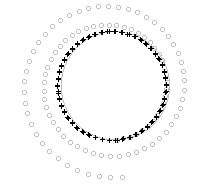} 
\end{minipage}

\end{tabular}
\caption{Vertically rolling disk with $\alpha=1/3, 1/2, 1$, repectively. \label{fig6}}
\end{figure}

\subsection{The knife edge}

As was the case with the vertically rolling disk, also the solutions of the knife edge form a circular path in the $(x,y)$-plane. Continuing the analogue with the previous example, the nonholonomic integration (in grey with circle symbols) results in a spiral, while the variational integration (in black with crosses) follows more closely the circular path, see figure~\ref{knife}.
\begin{figure}[h]\begin{center}
\includegraphics[scale=0.6]{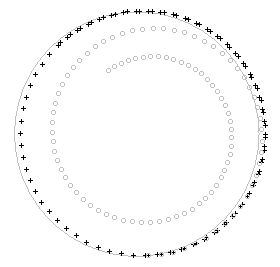}\end{center} \caption{The knife edge with $\alpha=0$ \label{knife}} 
\end{figure}

\subsection{The nonholonomic particle}

The function
\[
\tilde L = \frac{1}{2} {\dot x}^2+ \frac{\sqrt{1+x^2}}{2} \left(\frac{{\dot y}^2}{{\dot x}} +\frac{{\dot z}^2}{x{\dot x}}\right)
\]
is a free Lagrangian for the nonholonomic particle.
In each of the figures~\ref{fig7} and \ref{fig8} the dashed black curve represents the exact solution, the thick black the variational solution and the thick grey the nonholonomic solution. The figures show that both the variational method and the nonholonomic one do not give very accurate solutions. However, changing the parameter $\alpha$ does not seem to affect the variational solution as much as it does the nonholonomic one. Indeed, the variational solution remains more or less of the same accuracy for the different $\alpha$-values. On the other hand, the nonholonomic solution can be made more or less accurate by changing $\alpha$. It seems that the best accuracy is reached somewhere in the neighbourhood of $\alpha = 1/3$, but how could one have guessed this beforehand? Remark also that this value is not same as the the best choice we had found for the nonholonomic integrator of the vertically rolling disk (where $\alpha = 1/2$ gave the best accuracy).

\begin{figure}[h]

 \begin{center}

\includegraphics[scale=0.175]{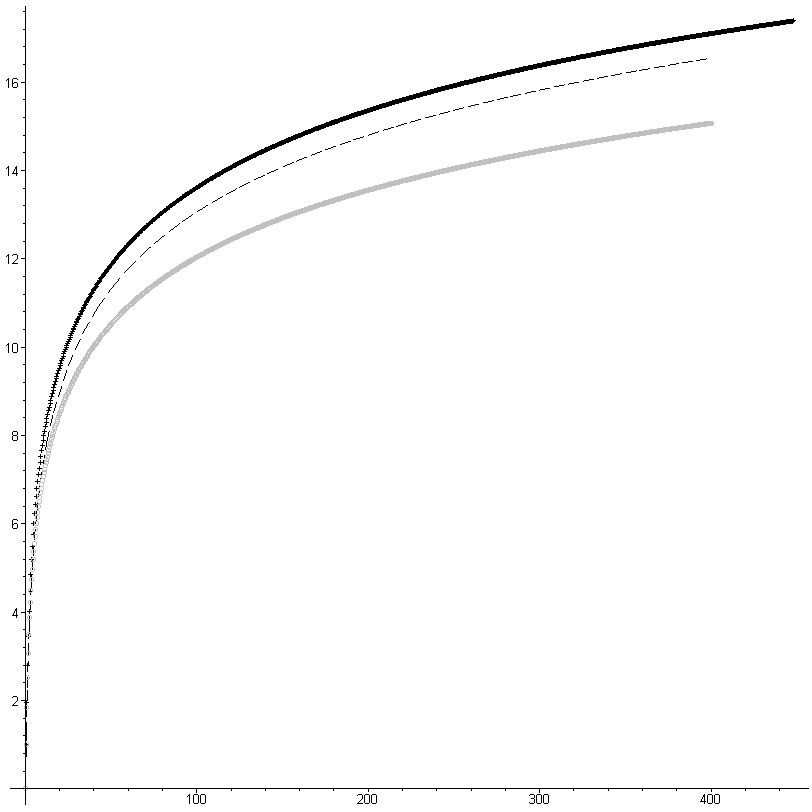} 

\end{center}

\begin{tabular}{cc}
\begin{minipage}{3.9cm}

\includegraphics[scale=0.175]{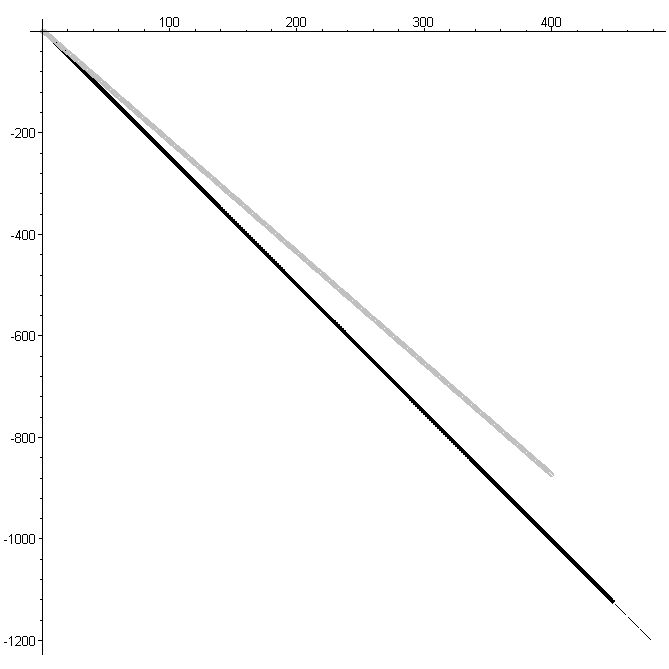} 
\end{minipage} & \begin{minipage}{3.9cm}

\includegraphics[scale=0.175]{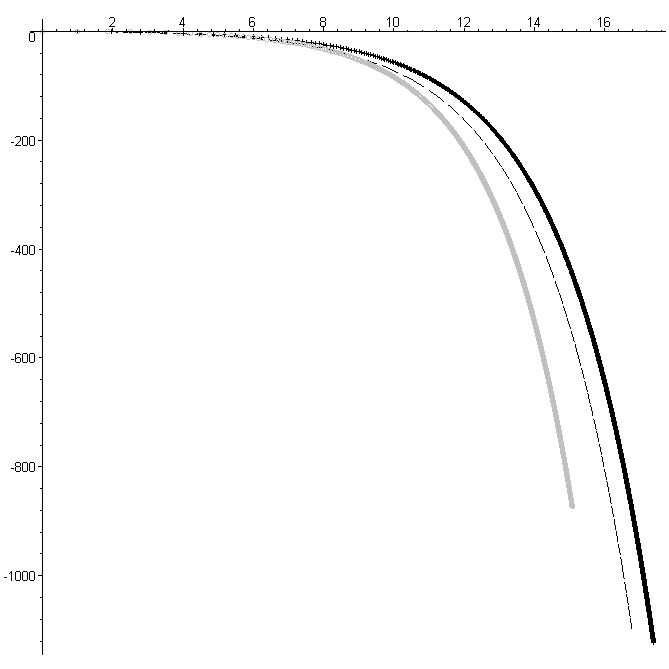} 
\end{minipage}
\end{tabular}
\caption{The nonholonomic particle: $xy$-, $xz$- and $yz$-solution with $\alpha=0$.\label{fig7}}
\end{figure}

\begin{figure}[h]
\begin{tabular}{cc}
\begin{minipage}{3.9cm}

\includegraphics[scale=0.15]{fig9a} 

\end{minipage}
& \begin{minipage}{3.9cm}

\includegraphics[scale=0.15]{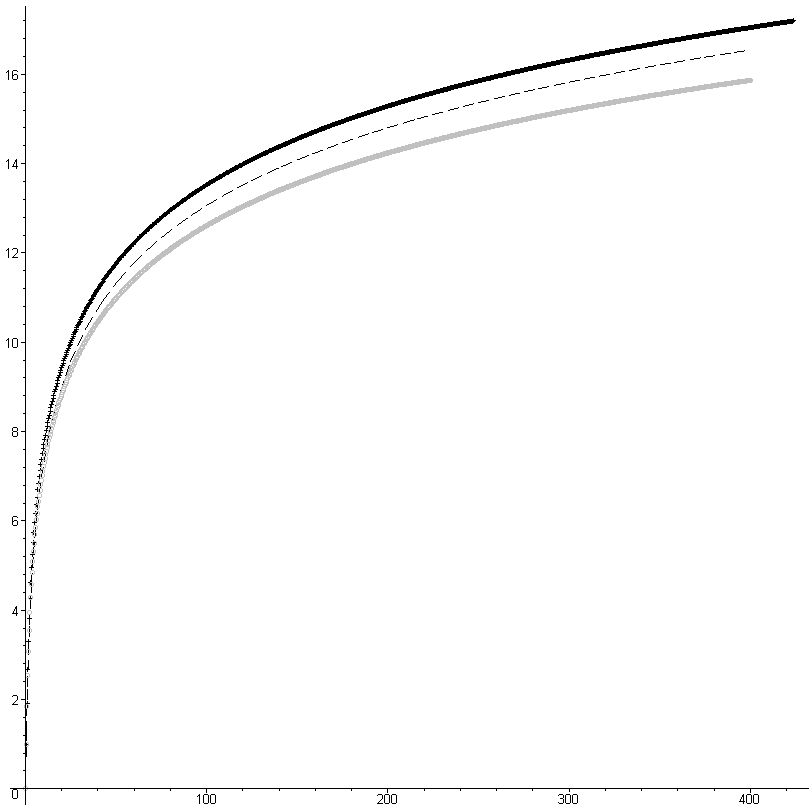} 
\end{minipage} \\  \begin{minipage}{3.9cm}

\includegraphics[scale=0.15]{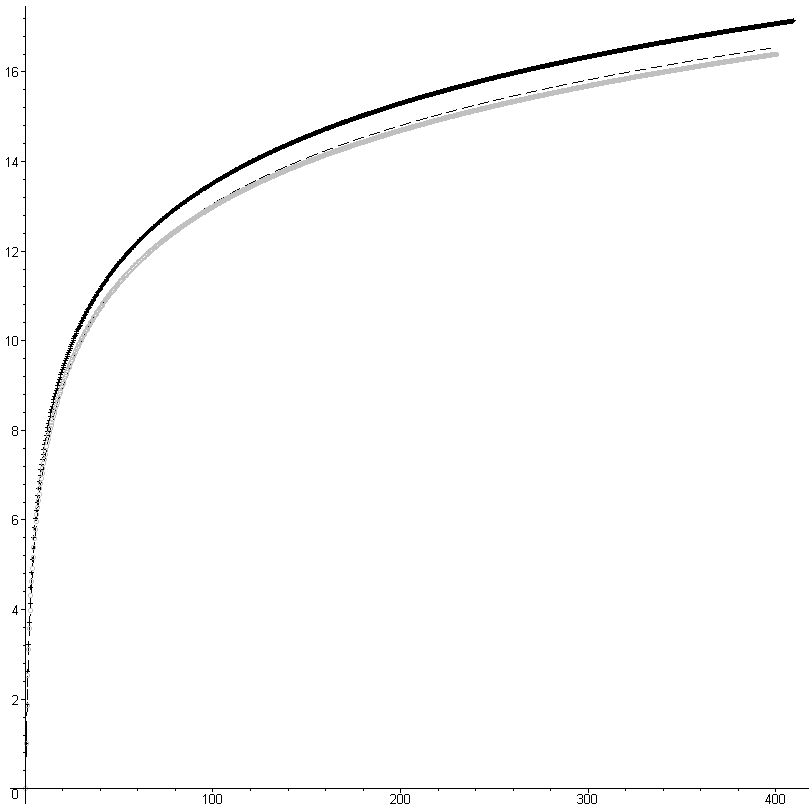} 
\end{minipage}

&

\begin{minipage}{3.9cm}

\includegraphics[scale=0.15]{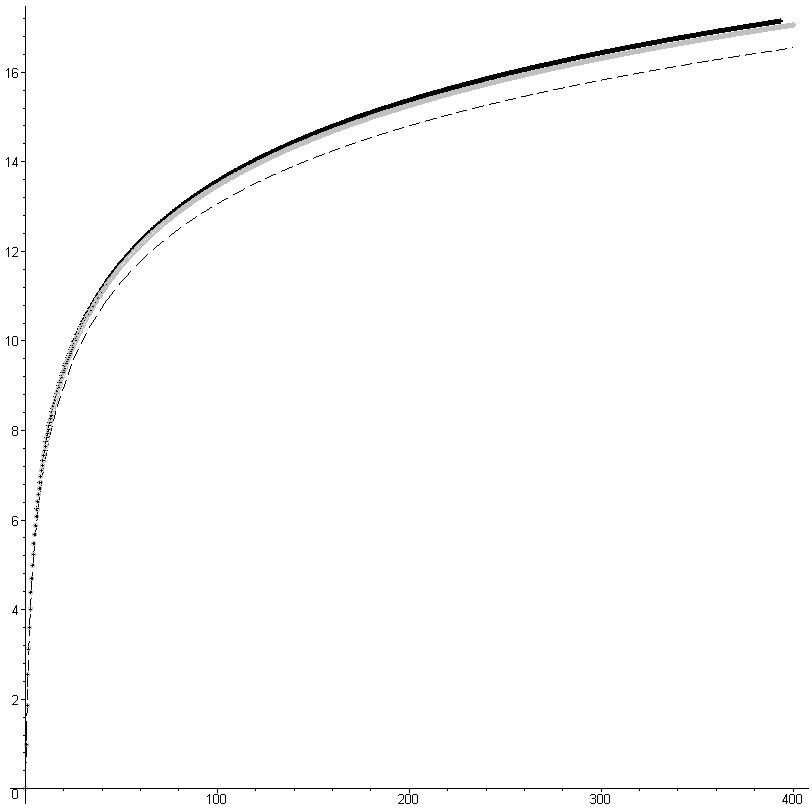} 

\end{minipage}
\\ \begin{minipage}{3.9cm}

\includegraphics[scale=0.15]{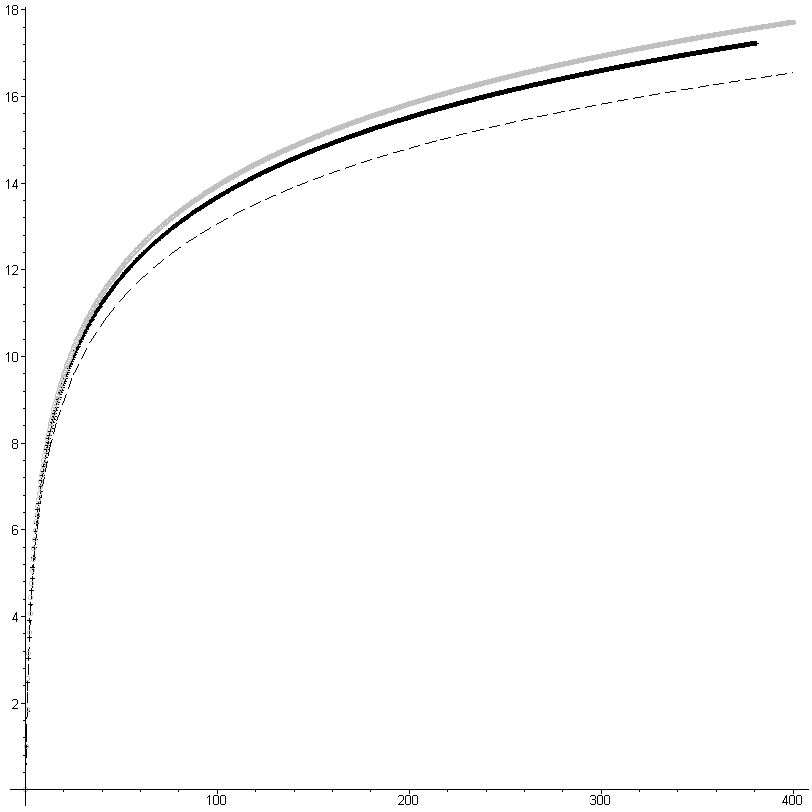} 
\end{minipage} & \begin{minipage}{3.9cm}

\includegraphics[scale=0.15]{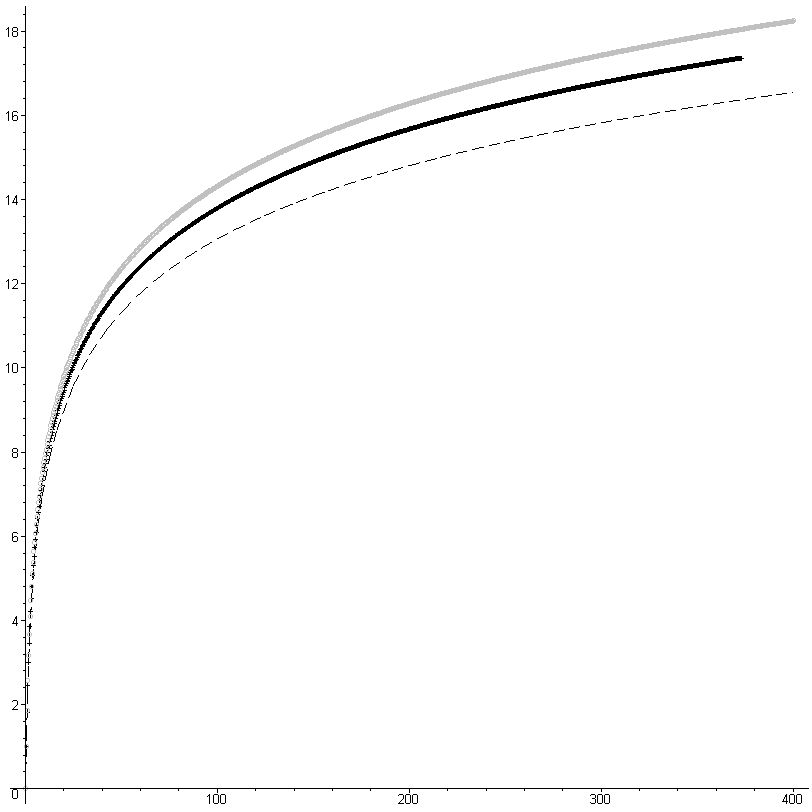} 
\end{minipage}

\end{tabular}

\caption{The nonholonomic particle: $xy$-solution with $\alpha=0, 1/5, 1/3, 1/2, 2/3, 4/5$, respectively.\label{fig8}}
\end{figure}

\subsection{Preliminary conclusion}

In each of the discussed examples the variational integrator (with one of the  Lagrangians (\ref{Lag1}) and (\ref{Lag2}) of proposition~\ref{Prop1}) seemed to give better results than the known nonholonomic integrators. Unlike the outcome for the nonholonomic integrator, the results for the variational integrator seemed to be independent of or, at least, stable under changing the parameter $\alpha$. Needless to say, the results above are, of course, very partial and are they are only intended to motivate further investigation on this topic. For example, we need to check if more involved discretization procedures, such as the ones mentioned at the end of section~\ref{sec3}A demonstrate the same behaviour as the one we have encountered so far.

\section{Further systems}\label{sec4}

The class of nonholonomic systems treated above is very restricted. The reason is, of course, that the search for a solution of the inverse problem of the calculus of variations (in the proof of proposition~\ref{Prop1}) is too hard and too technical to be treated in the full generality of a nonholonomic systems with an arbitrary given Lagrangian and arbitrary given constraints. Also, since there are infinitely many possible choices for the associated systems, it is not clear from the outset which one of them will be variational, if any.

For these reasons, future extensions of the obtained results will strongly depend on well-chosen particular new examples. For example, we could try to find a free Lagrangian for a nonholonomic system with a potential of the form $V(r_2)$. Typical examples of such systems are the mobile robot with a fixed orientation\[\begin{array}{l}
 L= \onehalf m ({\dot x}^2 + {\dot y}^2) + \onehalf I
{\dot\theta}^2 + \frac{3}{2}J {\dot\psi}^2 - 10 \sin\psi, \\[2mm]
\dot x = R\cos\theta\dot\psi,\quad \dot y = R\sin\theta\dot\psi, \end{array} \] (an example that also appears in the paper \cite{CortesMartinez}) or the knife edge on an inclined plane, where \[ L= \onehalf m ({\dot x}^2
+ {\dot y}^2) + \onehalf J {\dot\phi}^2 + mg x \sin\alpha, \quad
\dot x \sin \phi = \dot y \cos\phi. \]
In more general terms, such systems have a Lagrangian of the form
  \[ L=\frac{1}{2}(I_1{\dot r_{1}}^2+I_2{\dot r_{2}}^2 +
I_3 {\dot s}^2) -  V(r_2)\]
 and a constraint of the form \[ {\dot s} =-A (r_{1})\dot r_{2},\] and we can
we can consider associated second-order equations, in a way that is analogous to the way we arrived at the second system (\ref{second}) before: They are now of the form
\begin{equation}\label{extension} \left\{ \begin{array}{l}
 \ddot r_{1} =0, \\[2mm] \ddot r_{2} = \Gamma_2 (r_1) \dot r_{1}\dot
 r_{2} + t_2 (r_1,r_2),\\[2mm]
 {\ddot s} =  \Gamma_3(r_1) \dot r_{1} {\dot s}  +
 t_3(r_1,r_2). \end{array} \right.
 \end{equation}
 Remark that compared to the equations (\ref{second}), the presence of the extra potential brings the terms $t_i(r_1,r_2)$ into the picture. A first result is the following.
 \begin{propo} {\em There does not exists a regular Lagrangian for the second order systems (\ref{extension}).}
\end{propo}

\begin{proof}
As before, the proof follows from a careful analysis of the algebraic conditions which can be derived from the Helmholtz conditions.
\end{proof}

For systems with more than one constraint, the result is still open. Remark that the proposition does not exclude the existence of an other variational `associated' system.

\section*{Acknowledgments}
TM acknowledges a Marie Curie Fellowship within the
 6th European Community Framework Programme and a postdoctoral fellowship of the  Research Foundation - Flanders. The research of AMB and OEF was supported in part by the Rackham
Graduate School of the University of Michigan, through the Rackham
Science award, and through NSF grants DMS-0604307 and CMS-0408542.

\nocite{*}
\bibliographystyle{natcong}

%

%

\end{document}